\documentclass[11pt, 
]{article}

\title{Fully Dynamic Spectral Sparsification for Directed Hypergraphs}

\usepackage{authblk}

\author[1]{Sebastian Forster}
\author[2]{Gramoz Goranci}
\author[3]{Ali Momeni}
\affil[1]{Department of Computer Science, University of Salzburg, Austria}
\affil[2]{Faculty of Computer Science, University of Vienna, Austria}
\affil[3]{Faculty of Computer Science, UniVie Doctoral School Computer Science DoCS, University of Vienna, Austria}
\date{}

\usepackage{graphicx}
\usepackage{subcaption}

\graphicspath{{graphics/}{./graphics/}}

\usepackage{enumitem}
\setlist[itemize]{itemsep=.5pt, topsep=3pt,leftmargin=*}
\setlist[enumerate]{itemsep=.5pt, topsep=3pt,leftmargin=*}
\usepackage{titlesec}
\titlespacing{\paragraph}{0pt}{1em}{1em} %

\usepackage[
    dvipsnames   %
        ]{xcolor}

\definecolor{myOrange}{RGB}{230, 159, 0}
\definecolor{myLightBlue}{RGB}{86, 180, 233}
\definecolor{myGreen}{RGB}{0, 158, 115}
\definecolor{myYellow}{RGB}{240, 228, 66}
\definecolor{myDarkBlue}{RGB}{0, 114, 178}
\definecolor{myRed}{RGB}{213, 94, 0}
\definecolor{myPink}{RGB}{204, 121, 167}

\usepackage{amsmath}
\usepackage{amsthm}  %
\usepackage{amssymb}  %
\usepackage{mathtools}
\usepackage{enumitem}
\usepackage[]{bbold}

\usepackage{thmtools}
\usepackage{thm-restate}

\newtheorem{theorem}{Theorem}[section]

\newtheorem{lemma}[theorem]{Lemma}

\newtheorem{claim}{Claim}[subsection]

\usepackage[linesnumbered, boxruled, vlined]{algorithm2e}

\DontPrintSemicolon

\SetKwInput{Maintain}{Maintain}

\SetKwProg{Procedure}{Procedure}{}{}

\SetArgSty{textnormal}

\usepackage{xspace}

\usepackage{mleftright}

\newcommand{\vect}[1]{\ensuremath{\boldsymbol{#1}}}
\newcommand{\atmost}[1]{\ensuremath{O(#1)}\xspace}

\newcommand{\poly}[1]{\ensuremath{\operatorname{poly}\mleft( #1 \mright)}\xspace}

\def\Htil{\widetilde H}
\def\Otil{\widetilde O}
\def\Itil{\widetilde I}

\usepackage{tcolorbox}

\newcommand{\sparsifier}[1]{\ensuremath{\widetilde #1}\xspace}
\newcommand{\hypergraph}[1]{{\ensuremath{ #1}\xspace}}

\newcommand{\SpectralHypersparsifier}[1]{\ensuremath{ (1 \pm \varepsilon) }-spectral hypersparsifier\xspace}

\newcommand{\head}[1]{\ensuremath{\operatorname{h}\mleft( #1 \mright)}\xspace}

\newcommand{\tail}[1]{\ensuremath{\operatorname{t}\mleft( #1 \mright)}\xspace}

\newcommand{\E}[1]{\ensuremath{E\mleft( #1 \mright)}\xspace}

\newcommand{\C}[1]{\ensuremath{C_{#1}}\xspace}

\renewcommand{\S}[1]{\ensuremath{S_{#1}}\xspace}

\newcommand{\hh}[1]{\ensuremath{H_{#1}}\xspace}

\newcommand{\ConstantLambda}{\ensuremath{c_{ \lambda }}\xspace}

\newcommand{\ilast}{{\ensuremath{i_{\text{last} }}}\xspace}

\usepackage[
    backend=biber,    %
    maxcitenames=10,    %
    maxbibnames=10,    %
    maxalphanames=10,    
    maxnames=10,    %
    style=alphabetic,    %
    sorting=anyt,    %
    doi=false,    %
    isbn=false,    %
    eprint=false,    %
    url=false,    %
	backref=true,    %
	backrefstyle=none,    %
        ]{biblatex}

\AtEveryBibitem{
\ifentrytype{inproceedings}{
\clearname{editor}%
\clearfield{valume}%
\clearfield{number}%
\clearfield{series}%
\clearlist{publisher}%
\clearfield{address}%
\clearlist{location}%
\clearfield{month}%
\clearfield{eprint}%
\clearfield{isbn}%
\clearfield{issn}%
}{}
}

\newbibmacro{string+doi}[1]{%
    \iffieldundef{doi}{%
        \iffieldundef{url}{#1}{\href{\thefield{url}}{#1}}
        }{%
        \href{http://dx.doi.org/\thefield{doi}}{#1}
        }
    }

\DeclareFieldFormat*{title}{%
    \usebibmacro{string+doi}{\mkbibemph{#1}}
    }

\renewbibmacro{in:}{%
    \ifentrytype{article}{}{\printtext{\bibstring{in}\intitlepunct}}}

\newtoggle{bbx:datemissing}

\renewbibmacro*{date}{\toggletrue{bbx:datemissing}%
}

\renewbibmacro*{issue+date}{}%

\newbibmacro*{date:print}{%
    \setunit{\addcomma\addspace}    %
    \togglefalse{bbx:datemissing}%
    \printdate}

\renewbibmacro*{chapter+pages}{%
    \printfield{chapter}%
    \setunit{\bibpagespunct}%
    \printfield{pages}%
    \newunit
    \usebibmacro{date:print}%
    \newunit}

\renewbibmacro*{note+pages}{%
    \printfield{note}%
    \setunit{\bibpagespunct}%
    \printfield{pages}%
    \setunit{\addcomma\addspace}%
    \usebibmacro{date}%
    \newunit}

\renewbibmacro*{addendum+pubstate}{%
    \iftoggle{bbx:datemissing}
        {\usebibmacro{date:print}}
        {}%
    \printfield{addendum}%
    \newunit\newblock
    \printfield{pubstate}}

\appto{\bibsetup}{\sloppy}

\DefineBibliographyStrings{english}{%
    page={page},%
    pages={pages},%
    }

\DefineBibliographyStrings{english}{
  backrefpage = {},
  backrefpages = {},
}

\DeclareFieldFormat{bracketswithperiod}{\mkbibbrackets{#1}}

\renewbibmacro*{pageref}{%
  \iflistundef{pageref}
    {}
    {\printtext[bracketswithperiod]{%
       \ifnumgreater{\value{pageref}}{1}
         {\bibstring{backrefpages}}
         {\bibstring{backrefpage}}%
       \printlist[pageref][-\value{listtotal}]{pageref}}%
     }}

\usepackage{xurl}

\usepackage[
    breaklinks=true,    %
    colorlinks=true,    %
    linkcolor=myDarkBlue,    %
    urlcolor=myRed,    %
    citecolor=myGreen,    %
    bookmarks=true,    %
    bookmarksopenlevel=2,    %
    ]{hyperref}

\usepackage[nameinlink]{cleveref}

\crefname{subsection}{subsection}{subsections}

\crefname{problem}{problem}{problems}
\crefname{claim}{claim}{claims}
\crefname{algorithm}{algorithm}{algorithms}
\crefname{proof}{proof}{proofs}
\crefname{observation}{observation}{observations}
\crefname{invariant}{invariant}{invariants}

\usepackage[english]{babel}
\usepackage{fix-cm}
\usepackage[utf8]{inputenc}
\usepackage[T1]{fontenc}
\usepackage{amsfonts}
\usepackage{slantsc}
\usepackage{lmodern}

\usepackage[protrusion=true,expansion,kerning=true]{microtype}

\newcommand{\arxivVsConference}[2]{#1}

\begin{document}
\maketitle
\pagenumbering{roman}

\begin{abstract}
There has been a surge of interest in spectral hypergraph sparsification, a natural generalization of spectral sparsification for graphs.
In this paper, we present a simple fully dynamic algorithm for maintaining spectral hypergraph sparsifiers of \textit{directed} hypergraphs.
Our algorithm achieves a near-optimal size of $O(n^2 / \varepsilon ^2 \log ^7 m)$ and amortized update time of $O(r^2 \log ^3 m)$, where $n$ is the number of vertices, and $m$ and $r$ respectively upper bound the number of hyperedges and the rank of the hypergraph at any time. 

 We also extend our approach to the parallel batch-dynamic setting, where  a batch of any $k$ hyperedge insertions or deletions can be processed with $O(kr^2 \log ^3 m)$ amortized work and $O(\log ^2 m)$ depth. This constitutes the first spectral-based sparsification algorithm in this setting.
\end{abstract}

\newpage
\setcounter{page}{1}
\pagenumbering{arabic}
\section{Introduction}

Sparsification--the process of approximating a graph with another that has fewer edges while preserving a key property--is a central paradigm in the design of efficient graph algorithms.
A fundamental property of interest is graph cuts, 
which are not only foundational in graph theory and serve as duals to flows, 
but also have widespread applications in areas such as graph clustering~\cite{Lee:2014aa,Peng:2017aa,Feng:2018aa} and image segmentation~\cite{Boykov:2006aa,Kohli:2010aa}, to mention a few. 
In their seminal work, \arxivVsConference{\citeauthor{Benczur:1996aa}}{Bencz{\'{u}}r and Karger}~\cite{Benczur:1996aa} initiated the study of sparsifiers in the context of graph cuts. 
They showed that any graph admits a sparse reweighted cut-sparsifier
whilst paying a small loss in the approximation. 
\arxivVsConference{\citeauthor{Spielman:2011aa}}{Spielman and Teng}~\cite{Spielman:2011aa} introduced a variant of sparsification known as spectral sparsification, 
which generalizes cut sparsification and measures graph similarity 
via the spectrum of their Laplacian matrices. 
The development of such sparsification techniques has had a profound impact across algorithm design~\cite{Madry:2010aa,Sherman:2013aa,Spielman:2014aa,Peng:2016aa}, 
with the Laplacian paradigm standing out as a key example~\cite{Boman:2008aa,Spielman:2008aa,Daitch:2008aa,Kelner:2009aa,Koutis:2011aa}.

Motivated by the need to capture complex interdependencies in real-world data,
beyond the pairwise relationships modeled by traditional graphs, 
there has been a surge of interest in recent years in developing spectral sparsifiers for hypergraphs. 
These efforts have led to algorithmic constructions 
that achieve near-optimal size guarantees~\cite{Kapralov:2021aa,Kapralov:2021ab,Oko:2023aa,Jambulapati:2023aa,Lee:2023aa}. 
However, a common limitation of these algorithms is the assumption that the input graph is static--an assumption that does not hold in many practical settings. 
For example, real-world graphs, such as those modeling social networks, are inherently dynamic and undergo continual structural changes. 
For undirected hypergraphs, this limitation has been addressed in two recent independent works~\cite{Goranci:2025aa,Khanna:2025aa}, 
which show that spectral hypergraph sparsifiers can be maintained dynamically, 
supporting both hyperedge insertions and deletions in polylogarithmic time with respect to input parameters. 
This naturally leads to the question of whether \emph{directed} hypergraphs also admit similarly efficient dynamic sparsification algorithms.

In this paper, we study problems at the intersection of dynamic graph-based data structures and spectral sparsification for directed hypergraphs. 
Specifically, we consider the setting where a directed hypergraph undergoes hyperedge insertions and deletions, 
and the goal is to efficiently process these updates while maintaining a spectral sparsifier that approximates the input hypergraph. 
Our work builds upon variants of two key algorithmic constructions: the static spectral sparsification framework for directed hypergraphs developed by \arxivVsConference{\citeauthor{Oko:2023aa}}{Oko et al.}~\cite{Oko:2023aa}, 
and the dynamic sparsifier maintenance techniques for ordinary graphs by \arxivVsConference{\citeauthor{Abraham:2016ab}}{Abraham et al.}~\cite{Abraham:2016ab}. 
Leveraging insights from both lines of work and modifying them to our setting, 
we design efficient dynamic algorithms for maintaining spectral sparsifiers of directed hypergraphs, 
as formalized in the theorem below.

\begin{restatable}{theorem}{main} \label{th:main}
    Given a directed hypergraph \(H = (V, E, \vect{w})\)
    with \(n\) vertices, rank \(r\), and at most \(m\) hyperedges (at any time),
    there is a fully dynamic data structure 
    that, with high probability,
    maintains a \SpectralHypersparsifier{} \(\widetilde H\) of \(H\)
    of size \atmost{n^2 / \varepsilon ^2 \log ^7 m}
    in \atmost{r^2 \log ^3 m}  
    amortized update time.
\end{restatable}

The guarantees provided by the above theorem are nearly tight, for the following reasons: (1) even reading the hyperedges in a hypergraph of rank \(r\) requires $\Theta(r)$ time, 
which means any update time must inherently depend on $r$,
and (2) in the setting of directed graphs, 
it is folklore that a complete bipartite graph on $n$ vertices,
with all edges directed from one partition to the other, 
gives an $\Omega(n^2)$ lower bound on the size of any directed sparsifier. 
Moreover, \arxivVsConference{\citeauthor{Oko:2023aa}}{Oko et al.}~\cite{Oko:2023aa} established an even stronger lower bound, 
showing that any spectral sparsifier must also incur a $\Omega(\epsilon^{-1})$ dependence. 
The latter implies that the size of our dynamic sparsifier is optimal up to a factor of $\epsilon^{-1}$ and polylogarithmic terms.

Our algorithmic construction also extends naturally to the closely related batch-dynamic setting. 
In this model--similar to the fully dynamic setting--updates consist of hyperedge insertions and deletions, but are processed in batches, enabling the exploitation of parallelism. 
This approach is particularly well-suited for handling high-throughput update streams and may more accurately reflect 
how dynamic changes are managed in real-time systems. 
Our result in this model is formalized in the theorem below.

\begin{restatable}{theorem}{parallel} \label{th:parallel}
    Given a directed hypergraph \(H = (V, E, \vect{w})\)
    with \(n\) vertices, rank \(r\), and at most \(m\) hyperedges (at any time)
    undergoing batches of \(k\) hyperedge additions or deletions,
    there is a parallel fully dynamic data structure 
    that, with high probability,
    maintains a \SpectralHypersparsifier{} \(\widetilde H\) of \(H\)
    of size \atmost{n^2 / \varepsilon ^2 \log ^7 m}
    in \atmost{k r^2 \log ^3 m} amortized  work
    and \atmost{\log ^2 m} depth.
\end{restatable}


\arxivVsConference{}{
The data structure of \Cref{th:parallel},
along with its analysis,
is explained in \arxivVsConference{\Cref{sec:parallel}}{Appendix~\ref{sec:parallel}}.
}

\subsection{Related Work}
We briefly discuss the related work for spectral sparsification on both graphs and hypergraphs below.

\paragraph{Static Algorithms.}
Starting with \citeauthor{Spielman:2011aa}~\cite{Spielman:2011aa}, spectral sparsification has been extensively studed on graphs~\cite{Batson:2012aa,Kapralov:2012aa,Batson:2013aa,Zhu:2015aa,Koutis:2016ab,Lee:2017aa,Li:2018aa}.
Recently, the concept has been extended to undirected and directed hypergraphs~\cite{Soma:2019aa,Bansal:2019aa,Kapralov:2021aa,Kapralov:2021ab,Rafiey:2022aa,Oko:2023aa,Lee:2023aa,Jambulapati:2023aa,Khanna:2024aa}.

\paragraph{Dynamic Algorithms.}
For undirected graphs, there have been several results for dynamic spectral sparsifiers that use a similar approach to ours.
This includes the work of \citeauthor{Abraham:2016ab}~\cite{Abraham:2016ab}, which achieves polylogarithmic update time using $t$-bundle spanners, and its extension against an adaptive adversary by \citeauthor{Bernstein:2022aa}~\cite{Bernstein:2022aa} and to directed graphs by \citeauthor{Zhao:2025aa}~\cite{Zhao:2025aa}.
More recently, \citeauthor{Khanna:2025aa}~\cite{Khanna:2025aa} achieved a fully dynamic spectral sparsifier with near-optimal size and update time for undirected hypergraphs.
A closely related notion of sparsification, namely vertex sparsifiers, has also been studied in the dynamic setting~\cite{Goranci:2018aa,Durfee:2019aa,Chen:2020ac,Gao:2021aa,Axiotis:2021aa,Brand:2022aa,Dong:2022aa}.

\paragraph{Distributed and Parallel Algorithms.}
\citeauthor{Koutis:2016aa}~\cite{Koutis:2016aa} achieved simple algorithms for spectral graph sparsification that can be implemented in many computational models, including both parallel and distributed settings, with sub-optimal guarantees on the sparsifier size.
Very recently, \citeauthor{Ghaffari:2025aa}~\cite{Ghaffari:2025aa} developed a parallel batch-dynamic algorithm for spanners.
Other related works include distributed vertex sparsifiers~\cite{Zhu:2021aa,Forster:2021aa}.

\paragraph{Online and Streaming Algorithms.}
For graphs, \citeauthor{Kelner:2013aa}~\cite{Kelner:2013aa} extended the sampling scheme based on effective resistances~\cite{Spielman:2008aa} to the semi-streaming setting.
\citeauthor{Cohen:2020aa}~\cite{Cohen:2020aa} obtained an online spectral sparsification algorithm for graphs.
Recently, \citeauthor{Soma:2024aa}~\cite{Soma:2024aa} proposed an online algorithm for spectral hypergraph sparsification.
For graphs in dynamic streams, \citeauthor{Ahn:2013aa}~\cite{Ahn:2013aa} developed a spectral sparsifier.
Also, there has been a series of work on spectral sparsification in dynamic streams for both graphs and hypergraphs~\cite{Kapralov:2017aa,Kapralov:2019aa,Kapralov:2020aa}.

\subsection{Technical Overview}
In this section, we present the main ideas behind our fully dynamic algorithm for maintaining a $(1\pm \varepsilon)$-spectral hypersparsifier $\Htil$ of a directed hypergraph $H$.
Our algorithm builds on the static algorithm of \cite{Oko:2023aa}, which we briefly review.

The algorithm of \cite{Oko:2023aa} constructs $\Htil$ by computing a sequence of hypergraphs \(H_1, \dots, H_k\), where \(H_1\) is a spectral hypersparsifier of \(H\), \(H_2\) is a spectral hypersparsifier of \(H_1\), and so on, until \(H_k = \widetilde H\).
Each $H_i$ is obtained via simple sampling scheme (discussed shortly) that guarantees, with high probability, $H_i$ is proportionally smaller than $H_{i-1}$.
As a result, the number of iterations is bounded by $k= O(\log m)$.
In constructing \(H_{i}\) from \(H_{i-1}\), the algorithm obtains two sub-hypergraphs of \(H_{i-1}\): the {\it coreset} hypergraph $C_i$ and the {\it sampled} hypergraph $S_i$.

At a high level, $C_i$ consists of a sufficient number of heavy-weight hyperedges of $H_{i-1}$ (to be specified shortly).
This ensures that, when hyperedges are sampled uniformly at random from the remaining hypergraph $H_{i-1}\setminus C_i$ to construct $S_i$, the resulting union $H_i := C_i\cup S_i$ forms a spectral hypersparsifier of $H_{i-1}$.
More precisely, $C_i$ is constructed as follows.
For every pair $(u,v)\in V\times V$, the algorithm selects the $O(\log^3 m/\varepsilon^2)$ heaviest hyperedges whose tail\footnote{In a directed hypergraph, each hyperedge $e$ is a pair $(\tail{e}, \head{e})$, with the tail $\tail{e}\subseteq V$ and the head $\head{e}\subseteq V$ reflecting the direction of $e$. See \Cref{sec:preliminaries} for further details.} contains $u$ and whose head contains $v$.\footnote{If multiple hyperedges have equal weight, the algorithm picks them arbitrarily.}
These hyperedges are then added to $C_i$.

To construct $\S{i}$, each hyperedge in \(H_{i-1} \setminus \C{i-1}\) is sampled independently with probability \(1/2\), and its weight is doubled.
They prove that, with high probability, this simple sampling scheme produces a $(1\pm \varepsilon)$-spectral hypersparsifier \(H_i = \C{i} \cup \S{i}\) of \(H_{i-1}\).
Moreover, with high probability, \(|E(H_i)| \leq 3|E(H_{i-1})|/4\) and so \(k = \atmost{\log m}\).
See \Cref{fig:OST-algorithm} for an illustration of their approach.

\begin{figure}[t]
\centering
\setkeys{Gin}{width=\linewidth}

\begin{subfigure}{0.4\textwidth}
\includegraphics[page=1]{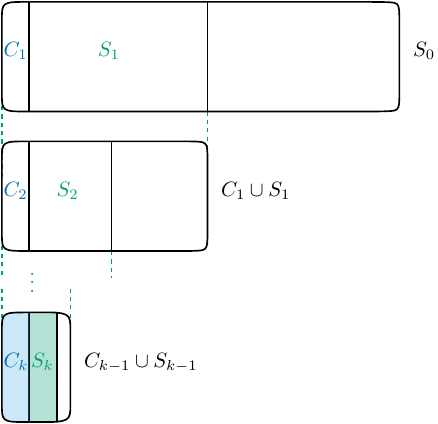}
\caption{}
\label{fig:OST-algorithm}
\end{subfigure}
\hfil
\begin{subfigure}{0.4\textwidth}
\includegraphics[page=2]{hypergraph.pdf}
\caption{}
\label{fig:our-algorithm}
\end{subfigure}
\caption{
Comparison of (a) the algorithm of \cite{Oko:2023aa} and (b) our static algorithm.
In each iteration \(i\), their algorithm recurses on \(\hh{i-1} = \C{i-1} \cup \S{i-1}\) and computes \(\hh{i} = \C{i} \cup \S{i}\) for the next iteration.
In contrast, our algorithm recurses solely on \S{i-1}, adds the coreset \C{i} to the sparsifier \sparsifier{H}, and computes the sampled hypergraph \S{i} for the next iteration.
After \(k= O(\log m)\) iterations, our algorithm terminates and returns \(\sparsifier{H} = \C{1} \cup \dots \C{k} \cup S_k\), whereas the algorithm of \cite{Oko:2023aa} returns \(\sparsifier{H} = \C{k} \cup \S{k}\) (the shaded parts in the figures).
The increase in the size of \sparsifier{H} in our algorithm allows us to maintain \sparsifier{H} efficiently in the dynamic setting, as detailed in \Cref{subsec:dynamic}.
}
\label{fig:algorithm}
\end{figure}

Unfortunately, the static algorithm cannot be directly converted into an efficient dynamic algorithm.
This is mostly due to the way the hypergraphs in the sequence \(H_1, \dots, H_k\) are built on top of each other, as a single change in $H$ can propagate into $O(m)$ changes across the sequence.
As an example, we explain how the removal of a hyperedge $e$ from $H_i$ can cause at least two changes in $H_{i+1}$, which can eventually lead to $\atmost{2^k}= \atmost{m}$ changes throughout the sequence.
Assume a hyperedge $e$ is removed from $H$ and that $e$ also belongs to the coreset hypergraph $C_i$ of $H_{i-1}$.
As $H_{i+1}$ is a sub-hypergraph of $H_i= C_i\cup S_i$, $e$ may also belong to $H_{i+1}$, necessitating its removal from $H_{i+1}$ as well.
To maintain $C_i$ as a coreset of $H_{i-1}$, the algorithm must replace $e$ with a next heaviest hyperedge $e'$ from $H_{i-1}\setminus C_i$.
If $e'$ is an unsampled hyperedge (i.e., belongs to $H_{i-1}\setminus S_i$), its addition to $C_i$ (and so $H_i$) may require updating $H_{i+1}$ as well: as $H_{i+1}$ is a sub-hypergraph of $H_i$, the (newly added) hyperedge $e'$ could be heavier than some hyperedge $e''$ in $C_{i+1}$, necessitating the replacement of $e''$ in $C_{i+1}$.
Thus, the removal of a hyperedge $e$ from $H_i$ can trigger at least two removals (namely, $e$ and $e''$) and one insertion (namely, $e'$) in $H_{i+1}$.
See \Cref{subsec:static} for further details.

To overcome this issue, our static algorithm deviates from that of \cite{Oko:2023aa} in the following way.
We ensure that each hyperedge is included in at most one coreset by recursing on the sampled hypergraph \S{i} rather than on $\C{i} \cup \S{i}$.
At the same time, we add \C{i} to the sparsifier $\Htil$.
i.e., we set $\Htil= C_1\cup \dots\cup C_k\cup S_k$, which can be verified to remain a spectral sparsifier of $H$ (as discussed in \Cref{lem:spectral-sparsify}).
See \Cref{fig:our-algorithm} for an illustration.
Using this approach, after the deletion of $e$ from $H_i$, each hypergraph $H_{i+1},\dots, H_k$ in the sequence will undergo at most one change.
Specifically, in the case of replacing a hyperedge $e'$ by $e$ as in the high-recourse example above, there are two possibilities: either (1) $e'$ does not belong to $H_{i+1}$ (i.e., it is not in the sampled hypergraph $S_i$), in which case the replacement does not affect $H_{i+1}$, or (2) $e'$ belongs to $H_{i+1}$, in which case the update is interpreted as the removal of $e'$ from $H_{i+1}$ (note that since $C_i$ is excluded from $H_{i+1}$, the hyperedge $e$ no longer belongs to $H_{i+1},\dots, H_k$).
The downside of this approach is that it increases the size of \(\widetilde H\) from \atmost{n^2/\varepsilon^2 \log ^3 (n/\varepsilon)} to \atmost{n^2/\varepsilon^2 \poly{\log m}}.
This increase becomes noticeable only when $m$ is exponential in $n$.
Even in that case, however, the size of $\Htil$ stays $\poly{n}$, which is asymptotically much smaller than the exponential size of $H$.

To dynamize our static algorithm, we adapt an approach similar to that of \cite{Abraham:2016ab} for graphs.
We first design a decremental data structure (\Cref{alg:decremental-spectral-sparsify} in \Cref{subsub:decremental}), where the updates consist of only hyperedge deletions, and then use it to design a fully dynamic data structure (\Cref{alg:reduction} in \Cref{subsub:fully-dynamic}) via a reduction technique.

Our decremental algorithm leverages the fact that a hyperedge deletion in $H$ causes at most one hyperedge deletion in each hypergraph in the sequence $H_1,\dots, H_k$.
The removal of a hyperedge $e$ from $H_i= C_i\cup S_i$ is handled using a straightforward replacement scheme.
If $e$ belongs to the sampled hypergraph $S_i$, then $S_i$ remains {\it valid} after the removal of $e$, in the sense that it is still a set of hyperedges sampled uniformly at random from the updated $H_i\setminus C_i$.
The update procedure is more involved if $e$ belongs to the coreset hypergraph $C_i$.
Recall that $e$ was added to $C_i$ through a pair $(u,v)\in V\times V$, where $u$ belongs to the tail of $e$ and $v$ belongs to the head of $e$.
The replacement of $e$ is handled by removing it from all sets representing such pairs and then selecting a hyperedge with high weight that is not already in $C_i$.
Since there are $O(r^2)$ such pairs, this results in an $\Otil(r^2)$ update time for maintaining $H_i$.
See \Cref{subsub:decremental} for more details.

To convert our decremental data structure to a fully dynamic one, we leverage the decomposability property of spectral sparsifiers: the union of spectral sparsifiers of hyperedge partitions of $H$ forms a spectral sparsifier of $H$ (see \Cref{lem:decomposability} for a formal statement).
The main idea is to maintain a hyperedge partition \(I_1, \dots, I_l\) of \(H\) where \(|E(I_i)| \leq 2^i\) at any time.
Deletions are handled by passing them to the respective \(I_i\), whereas insertions are more difficult to handle: the data structure finds an integer \(j\) and moves all the hyperedges in \(I_1,\dots,I_{j-1}\) to \(I_j\) along with the inserted hyperedge.
The choice of $j$ depends on the number of insertions so far, with smaller values of $j$ (corresponding to hypergraphs considerably smaller than $H$) being chosen more frequently.
On each $I_i$, we run our decremental data structure to maintain a spectral sparsifier $\Itil_i$ of $I_i$, and upon an insertion, we reinitialize it for $I_j$.
The algorithm sets $\Htil= \Itil_1\cup \dots\cup \Itil_l$, and since $k= O(\log m)$ the desired size of $\Htil$ follows.
See \Cref{subsub:fully-dynamic} for further discussion.

Our approach in designing a decremental data structure and extending it to a fully dynamic one is similar to recent work on dynamic sparsification for undirected hypergraphs \cite{Goranci:2025aa,Khanna:2025aa}.
Moreover, \cite{Goranci:2025aa} also builds on the framework of \cite{Oko:2023aa}.
The key difference is that we adapt the directed framework of \cite{Oko:2023aa} (specifically, $\lambda$-coresets), whereas \cite{Goranci:2025aa} employs their undirected framework (specifically, $t$-bundle hyperspanners, a concept also used in \cite{Khanna:2025aa}).
Both \cite{Goranci:2025aa} and \cite{Khanna:2025aa} rely on spanner-based techniques to bound the effective resistances in the {\it associated graph} of the hypergraph as a crucial step in enabling their simple sampling scheme.
For the directed case, however, \cite{Oko:2023aa} shows that this translates to using coresets, which are structurally simpler than hyperspanners.
This structural simplicity allows us to design relatively simpler algorithms compared to those in \cite{Goranci:2025aa,Khanna:2025aa}.

It is noteworthy to mention the recent work of \cite{Khanna:2024aa} that reduces directed hypergraph sparsification to undirected hypergraph sparsification.
Combining this reduction with the fully dynamic undirected hypergraph sparsification result of \cite{Khanna:2025aa} yields a dynamic algorithm for directed hypergraph sparsification with guarantees similar to ours. 
However, the algorithm of \cite{Khanna:2025aa} is substantially more involved: (i) it relies on vertex-sampling steps, which our approach does not require, and (ii) it uses dynamic graph spanner constructions in a black-box manner. 
Moreover, it does not seem straightforward to extend their algorithm to the batch-parallel setting.
In comparison, our coreset-based construction is significantly simpler and potentially more practically relevant. 
It readily extends to the batch-parallel setting and achieves a better and explicit polylogarithmic overhead in both sparsifier size and update time.

Lastly, in \Cref{sec:parallel}, we show how to parallelize our data structures when \(H\) undergoes batches of \(k\) hyperedge deletions or additions as a single update.
This discussion also explains how to adapt our fully dynamic data structure to support batch updates rather than single updates.

\section{Preliminaries} \label{sec:preliminaries}

\paragraph*{Hypergraphs.}
A hypergraph $H = (V, E, \vect{w})$ consists of a set of vertices $V$, a set of hyperedges \(E\), and a weight vector $\vect{w}\in \mathbb R^{|E|}_+$.
The direction of a hyperedge $e \in E$ is defined by sets $\tail{e}$ and $\head{e}$ as follows.
Each hyperedge $e \in E$ is a pair $(\tail{e}, \head{e})$, where both $\tail{e}$ and $\head{e}$ are non-empty subsets of $V$.
The sets $\tail{e}$ and $\head{e}$ are called the {\it tail} and {\it head} of $e$, respectively; they indicate the direction of $e$.
Note that $\tail{e}$ and $\head{e}$ may overlap.

We use \(E(H)\) to denote the set \(E\) of hyperedges of \(H\)
whenever necessary, to avoid possible confusion.
We define \(n = |V|\) and \(m = |E|\) 
and call \(H\) an \(m\)-edge \(n\)-vertex hypergraph.
We say \(H\) is of rank \(r\)
if, for every hyperedge \(e \in E\),
\(|\tail{e} \cup \head{e}| \leq r\).

Given a vector \(\vect{x} \in \mathbb R^n\) defined on the set of vertices \(V\),
we define \(x_v\) 
to be the value of vector \vect{x}
at the element associated with vertex \(v \in V\).
Similarly,
for a vector \(\vect{w} \in \mathbb R^m _+\) 
defined on the set of hyperedges \(E\),
we define \(w_e\) 
to be the value of vector \vect{w}
at the element associated with hyperedge \(e \in E\). 

\paragraph*{Spectral Sparsification of Directed Hypergraphs.}
We define the spectral property of a directed hypergraph \(H = (V, E, \vect{w})\)
using the energy function defined in the following.
For a vector \(\vect{x} \in \mathbb R^n\),
we define the energy of \vect{x} with respect to \(H\) as
\begin{equation*}
    Q_{H}(\vect{x}) = \sum _{e \in E} w_e \max_{u \in \tail{e}, v \in \head{e}} \left( x_u - x_v \right)_+^2,
\end{equation*}
where \(( x_u - x_v )_+ = \max\{x_u - x_v, 0 \}\) and \((x_u - x_v)_+^2 = \left( (x_u - x_v)_+ \right)^2\).
Note that this definition generalizes a similar definition for graphs, given by $Q_G(\vect{x})= \sum _{uv \in E} w_{uv} ( x_u - x_v )^2$, which represents the total energy dissipated in $G$ when viewed as an electrical network.
In such electrical network, the endpoints of each edge $uv$ have potentials $x_u$ and $x_v$, respectively, and the edge itself has resistance $1/w_{uv}$.
This interpretation is closely related to the notion of electrical flows, a concept that is deeply intertwined with spectral analysis.

The central object of this paper 
is the notion of a spectral hypersparsifier.
A hypergraph \(\widetilde H = (V, \widetilde E, \vect{\widetilde w})\) is called a \SpectralHypersparsifier{} of \(H\)
if for every vector \(\vect{x} \in \mathbb R^n\),
\begin{equation*}
    (1 - \varepsilon) Q_{\widetilde H}(\vect{x}) 
    \leq  
    Q_{H}(\vect{x}) 
    \leq 
    (1 + \varepsilon) Q_{\widetilde H}(\vect{x}).
\end{equation*}

The following lemma will be useful later in proving the guarantees of our algorithm.

\begin{lemma}[Decomposability] \label{lem:decomposability}
    Let \(H_1, \dots, H_k\) partition the hyperedges of a hypergraph \(H\).
    For each \(1 \leq i \leq k\),
    let \(\widetilde H_i\) be a \((1 \pm \varepsilon)\)-spectral hypersparsifiers of \(H_i\).
    Then, 
    the union \(\bigcup _{l = 1} ^k \widetilde H_l\) is a \SpectralHypersparsifier{} of \(H\).
\end{lemma}

\begin{proof}
    By definition,
    for every vector \(\vect{x} \in \mathbb R^n\),
    we have
    \begin{equation*}
        (1 - \varepsilon) Q_{\widetilde H_i}(\vect{x})
        \leq
        Q_{H_i}(\vect{x})
        \leq
        (1 + \varepsilon) Q_{\widetilde H_i}(\vect{x}).
    \end{equation*}
    Summing over all \(1 \leq i \leq k\),
    results in
    \begin{equation*}
        (1 - \varepsilon) \sum_{i = 1} ^k Q_{\widetilde H_i}(\vect{x})
        \leq
        \sum_{i = 1} ^k Q_{H_i}(\vect{x})
        \leq
        (1 + \varepsilon) \sum_{i = 1} ^k Q_{\widetilde H_i}(\vect{x}),
    \end{equation*}
    which means
    \begin{equation*}
        (1 - \varepsilon) Q_{\widetilde H}(\vect{x})
        \leq
        Q_{H}(\vect{x})
        \leq
        (1 + \varepsilon) Q_{\widetilde H}(\vect{x})
    \end{equation*}
    as \(H_1, \dots, H_k\) partition \(H\).
\end{proof}

\paragraph*{Chernoff Bound \cite{Chernoff:1952aa,Mitzenmacher:2017aa}.}
Let \(X_1, \dots, X_k\) be independent random variables, 
where each \(X_i\) equals \(1\) with probability \(p_i\), and \(0\) otherwise.
Let \(X = \sum_{i = 1} ^k X_i\)
and \(\mu = \mathbb E[X] = \sum_{i = 1} ^k p_i\).
Then, for all \(\delta \geq 0\),
\begin{equation} \label{eq:chernoff}
    \mathbb P \left[ X \geq (1 + \delta) \mu \right] \leq \exp\left( - \frac{\delta ^2 \mu}{2 + \delta}  \right).
\end{equation}

\paragraph*{Parallel Batch-Dynamic Model.}
We use the work-depth model
to analyze our parallel algorithm.
Work is defined as the total number of operations done by the algorithm,
and depth is the length of the longest chain of dependencies.
Intuitively,
work measures the time required for the algorithm to run on a single processor,
whereas depth measures the optimal time
assuming the algorithm has access to an unlimited number of processors.

In the batch-dynamic setting,
each update consists of a batch of \(k\) insertions or deletions,
and the goal is to take advantage of performing several hyperedge insertions or deletions as a single update
to improve the work and depth of the parallel algorithm.
Note that this model is equivalent to the one with mixed updates (i.e., when the batch consists of both insertions and deletions),
as this can be transformed into two steps, each consisting of insertions and deletions separately, without asymptotically increasing the work or depth.

\section{Dynamic Spectral Sparsification} \label{sec:algorithm}

In this section, we present our fully dynamic algorithm of \Cref{th:main} for maintaining a \SpectralHypersparsifier{} of a directed hypergraph \(H = (V, E, \vect{w})\).
To achieve this goal, we use a static algorithm as the cornerstone for our dynamic algorithm.
The static algorithm is explained in \Cref{subsec:static} and is followed by the dynamic data structure in \Cref{subsec:dynamic}.

\subsection{The Static Algorithm} \label{subsec:static}
We start by briefly explaining the algorithm of \cite{Oko:2023aa}, which serves as a foundation for our algorithm.
Their algorithm constructs a $(1\pm \varepsilon)$-spectral hypersparsifier $\Htil$ of $H$ using an iterative approach: starting with \(\hh{0} = H\), each iteration \(i\) computes a sub-hypergraph \hh{i} of \hh{i-1}, until the final iteration computes $H_\ilast$, where $\Htil = H_\ilast$.

To compute \hh{i} from \hh{i-1}, the algorithm first obtains a \textit{$\lambda$-coreset} $C_i$ of \hh{i-1}.
Roughly speaking, $C_i$ is a sub-hypergraph containing `heavyweight' hyperedges of \hh{i-1} and is defined as follows.
The algorithm defines an ordering on hyperedges in $H$ by their weights in decreasing order.
Note that this ordering naturally extends to every hypergraph $H_i$ as a sub-hypergraph of $H$.
To construct $C_i$, the algorithm examines every pair $(u,v)\in V\times V$ and selects the first $\lambda$ hyperedges in the ordering (if any) that are not already in $C_i$, whose tail contains $u$ and whose head contains $v$.
These hyperedges are then added to $C_i$.
The second building-block of \hh{i-1} is the sampled hypergraph \S{i} of $H_{i-1}$ defined as follows.
The algorithm samples the non-heavy hyperedges (i.e., the ones in $\hh{i-1}\setminus \C{i}$) with probability \(1/2\) and adds them to \S{i} while doubling their weight.
Consequently, $\hh{i} = \C{i} \cup \S{i}$.
See \Cref{alg:coreset-and-sample} for a pseudocode.

\begin{algorithm}[t]
\KwIn{an \(m\)-edge hypergraph \(H = (V, E, \vect{w})\) and \( 0 < \varepsilon < 1 \)}
\KwOut{a coreset \C{} of \(H\) and a sampled hypergraph \S{}}
\( \lambda \gets \ConstantLambda \log ^3 m / \varepsilon ^2 \)
\tcc*{\ConstantLambda is a sufficiently large constant}
\( \C{}, \S{} \gets (V, \emptyset, \vect{0}) \)\;
Order the hyperedges of $H$ by their weights in decreasing order 

\ForEach(\tcc*[f]{compute \(\C{}\)}){pair \( (u, v) \in V \times V \)}{
    Construct the set \(\E{u, v} \subseteq E\) such that \(e \in \E{u, v}\) iff \(u \in \tail{e}\) and \(v \in \head{e}\)\;
    Add the first $\lambda$ hyperedges (if any) in the ordering from $\E{u, v} \setminus \C{}$ to \C{} while retaining their weights
}

\ForEach(\tcc*[f]{compute \(\S{}\)}){hyperedge \(e\) of \(H \setminus \C{}\)}{
    With probability \( 1/2 \), add \( e \) to \(S\) and double its weight
}

\Return{\( (\C{}, \S{}) \)}
\caption{Coreset-And-Sample(\( H, \varepsilon \))}
\label{alg:coreset-and-sample}
\end{algorithm}

Having heavyweight hyperedges in \C{i} ensures that the simple sampling scheme used for constructing $S_i$ results in $H_i= C_i\cup S_i$, which, with high probability, is a $(1\pm \varepsilon)$-spectral hypersparsifier of \hh{i-1} \cite[Lemma 4.3]{Oko:2023aa}.
Due to the sampling scheme, $H_i$  is roughly half the size of \(\hh{i-1}\), which ensures the termination of the algorithm after \(\ilast = \atmost{\log m}\) iterations.
Using this sequence $H_1,\dots, H_\ilast$ of hypergraphs,
they achieve a sparsifier $\Htil$ with an almost optimal size of \atmost{n^2/\varepsilon^2 \log ^3 (n / \varepsilon)} \cite[Theorem 1.1]{Oko:2023aa}.
See \Cref{fig:OST-algorithm} for an illustration.

Unfortunately, employing the sequence of hypergraphs \(\hh{1}, \dots, \hh{\ilast}\) can result in $O(m)$ recourse (and consequently, update time) in the dynamic setting.
For example, consider the removal of a hyperedge \(e\) from $H$ that also belongs to the coreset \C{i} of \hh{i-1}.
In this scenario, to ensure \C{i} remains a $\lambda$-coreset and so $\hh{i}=\C{i} \cup \S{i}$ remains a $(1 \pm \varepsilon_i)$-spectral hypersparsifier of \hh{i-1}, we need to replace \(e\) with another hyperedge \(e'\) from $\hh{i-1} \setminus \C{i}$.
i.e., if $e$ was added to $C_i$ through the pair $(u,v)$, hyperedge $e'$ in $\hh{i-1} \setminus \C{i}$ is the next hyperedge in the ordering whose tail contains $u$ and whose head contains $v$.
Since $S_i$ is a set of sampled hyperedges uniformly at random, it may be the case that $e'$ does not belong to $S_i$, and therefore to none of $H_j$ for $j>i$.
Since $e'$ is added to $H_i$ after the update (as it now belongs to $C_i$), it must be taken into account in the update of $H_{i+1}$ as a sub-hypergraph of (newly updated) $H_i$.
But now, $e'$ may be heavier than another hyperedge $e''$ in $C_{i+i}$, which necessitates the replacement of $e''$ with $e'$.
Since $e$ can be present in $C_{i+1}$ as well, this means that the deletion of $e$ from $H_i$ can result in at least $2$ changes in $H_{i+1}$, and at least \(2^{k-i} = \atmost{m}\) changes in the sequence, which is too expensive to afford.

To alleviate this issue, our static algorithm deviates from that of \cite{Oko:2023aa} by recursing on \S{i} instead of $\C{i}\cup \S{i}$ as follows.
At each iteration $i$, the algorithm recurses on \S{i-1}, adds the coreset \C{i} of \S{i-1} to \sparsifier{H}, and samples the rest, $\S{i-1}\setminus \C{i}$, to obtain a smaller hypergraph \S{i} for the next iteration.
More precisely, the algorithm starts with $\S{0} = H$ and an empty sparsifier $\sparsifier{H}$.
In the first iteration, it computes the coreset \C{1} of \S{0}, and adds it to \sparsifier{H}.
The algorithm then samples the remaining hypergraph $\S{0} \setminus \C{1}$ to compute \S{1} by sampling every hyperedge in $\S{0} \setminus \C{1}$ with probability $1/2$ and doubling their weight.
Similar to \cite{Oko:2023aa}, with high probability, $\C{i}\cup \S{i}$ is a $(1 \pm \varepsilon_i)$-spectral hypersparsifier of \S{i-1} (\Cref{lem:coreset-and-sample}).
The algorithm then recurses on \S{1}, and so on.
Thus, our algorithm adds the sequence of coresets $\C{1}, \dots, \C{\ilast}$ to \sparsifier{H} while recursing on the sequence of hypergraphs $\S{1}, \dots, \S{\ilast}$.
In the last iteration \ilast, our algorithm adds \C{\ilast} as well as \S{\ilast} to \sparsifier{H}.
See \Cref{alg:spectral-sparsify} for a pseudocode and \Cref{fig:our-algorithm} for an illustration.

This technique, which is similar to the one used in \cite{Abraham:2016ab} for graphs, ensures that each hyperedge of \sparsifier{H} is associated with at most \textit{one} coreset.
This guarantees \atmost{\log m} hyperedge deletions from \(\S{1}, \dots, \S{\ilast}\) after a hyperedge deletion in \(H\) (\Cref{lem:decremental-spectral-sparsify}), but in return, results in a \poly{\log m} overhead in the size of \sparsifier{H} as we include the coresets \(\C{1}, \dots, \C{\ilast}\) in \sparsifier{H}.
i.e., our algorithm ensures that \sparsifier{H} is a desired sparsifier of size \atmost{n^2/\varepsilon^2 \poly{\log m}}, which is asymptotically much smaller than $H$ even when $m$ is exponential in $n$.

\begin{algorithm}[t]
\KwIn{hypergraph \hypergraph{H = (V, E, \vect{w})} and \( 0 < \varepsilon < 1 \)}
\KwOut{a \SpectralHypersparsifier{} \sparsifier{H} of \(H\)}
\( i \gets 0 \)\;
\( k \gets \lceil \log _{3/4} m \rceil \)\;
\(m^\star \gets n^2/\varepsilon^2 \log ^3m\)\;
\(\sparsifier{H} \gets (V, \emptyset, \vect{0})\)\;
\( \S{0} \gets H \)\;

\While(\tcc*[f]{\(c\) is from \Cref{lem:coreset-and-sample}}){\( i \leq k \) and \( |E(H_i)| \geq 32 c m^\star \)}{
    \( (\C{i+1}, \S{i+1}) \gets \textsc{Coreset-And-Sample(\( \S{i}, \varepsilon /(2k) \))} \)\;
    \(\sparsifier{H} \gets \sparsifier{H} \cup \C{i+1}\)\;
    \( i \gets i + 1 \)
}

\( \ilast \gets i \)\;
\(\sparsifier{H} \gets \sparsifier{H} \cup \S{\ilast}\)\;

\Return{\( \widetilde H \)}
\caption{Spectral-Sparsify(\(H, \varepsilon\))}
\label{alg:spectral-sparsify}
\end{algorithm}

The rest of this section examines the correctness of \Cref{alg:coreset-and-sample,alg:spectral-sparsify}.
\Cref{alg:coreset-and-sample}, used as a subroutine of \Cref{alg:spectral-sparsify}, computes a coreset and a sampled hypergraph of the input hypergraph.
The guarantees of \Cref{alg:coreset-and-sample} are stated in the following lemma.

\begin{lemma} \label{lem:coreset-and-sample}
    Let \( 0 < \varepsilon < 1 \) and let \(H = (V, E, \vect{w})\) be an \( m \)-edge \( n \)-vertex hypergraph.
    For any positive constant \(c \geq 1\), if \(m \geq c \log m\), then \Cref{alg:coreset-and-sample} returns a coreset \C{} and a sampled hypergraph \S{} such that \(\C{} \cup \S{}\) is a \SpectralHypersparsifier{} of \(H\) of size \atmost{m/2 + (2 c m \log m)^{1/2} + \lambda n^2} with probability at least \( 1 - 1/m^{c} \). 
\end{lemma}

\begin{proof}
The lemma is derived from \cite[Lemma 4.3]{Oko:2023aa}. 
The only difference is that we guarantee a probability of success of at least \(1 - 1/\poly{m}\), instead of \(1 - 1/\poly{n}\).
Thus, we only prove the claim about the probability using an argument similar to that in \cite[Lemma 4.4]{Oko:2023aa}.

\underline{High probability claim:}
For every hyperedge \(e\) in \(\hh{} \setminus \C{}\), we define the random variable \(X_e\) to be equal \(1\) if \(e\) is sampled to be in \S{}, and \(0\) otherwise.
Let \(X = \sum_{\text{\(e\) in \S{}}} X_e\).
Since each hyperedge in \S{} is sampled independently with probability \(1/2\),  we can use \Cref{eq:chernoff} and we have \(\mu = \mathbb E\left[ X \right] = 1/2 |\hh{} \setminus \C{} | \leq m/2\).
By substituting \(\delta = \left( 2c \log m / m \right)^{1/2} \leq 2\) in \Cref{eq:chernoff},
\begin{equation*}
    \mathbb P \left[ X \geq \left( 1 +   \left( \frac{8c \log m}{m} \right)^{1/2} \right) \frac{m}{2} \right] \leq \exp\left( -\frac{1}{4} \left( \frac{8c \log m}{m} \right) \frac{m}{2} \right) = \exp\left( -c \log m \right) =  \frac{1}{m^c},
\end{equation*}
or equivalently
\begin{equation} \label{eq:coreset-probability}
    \mathbb P \left[ X \leq \frac{m}{2} + \left( 2c m \log m \right)^{1/2}  \right] \geq 1-\frac{1}{m^c}.
\end{equation}

Since each pair \((u, v) \in V \times V\) adds at most \(\lambda\) hyperedges to \(S\), we have \(|S| = \atmost{\lambda n^2}\). 
Together with \Cref{eq:coreset-probability}, it follows that 
\(\widetilde H = \C{} \cup \S{}\) has size \atmost{m/2 + (2 c m \log m)^{1/2} + \lambda n^2} with probability at least \( 1 - 1/m^{c} \).
\end{proof}

The guarantees of \Cref{alg:spectral-sparsify} are stated in the following lemma.

\begin{lemma} \label{lem:spectral-sparsify}
Let \( 0 < \varepsilon < 1 \) and let \(H = (V, E, \vect{w})\) be an \( m \)-edge \( n \)-vertex hypergraph.
Then, with high probability, \Cref{alg:spectral-sparsify} returns a \SpectralHypersparsifier{} \( \widetilde H \) of \( H \) of size \atmost{n^2 / \varepsilon ^2 \log^6 m }.
\end{lemma}

\begin{proof}
We prove each guarantee separately below.
    
\underline{Approximation guarantee:}
We prove that 
\( H_l' = \bigcup_{j = 1} ^l C{j} \cup \S{l} \) is a \(\left( 1 \pm \varepsilon/(2k) \right)^l\)-spectral hypersparsifier of \(H\) by induction on \(l\).

If \(l = 1\), 
by \Cref{lem:coreset-and-sample}, 
\(\hh{1} = \C{1} \cup \S{1}\) is a \((1 \pm \varepsilon/(2k))\)-spectral hypersparsifier of \(H\).

Suppose that, for an integer \(l > 1\),
\(H_l'\) be a \(\left( 1 \pm \varepsilon/(2k) \right)^l\)-spectral hypersparsifier of \(H\).
To compute \(H_{l+1}'\), the algorithm %
computes a \((1 + \varepsilon/(2k))\)-spectral hypersparsifier \(\widetilde S_l = \C{l+1} \cup \S{l+1}\) of \S{l}.
Since \(\bigcup_{j=1} ^l \C{j} \cup \S{l}\) partition \(H_l'\), by \Cref{lem:decomposability}, \(H_{l+1}' = \bigcup_{j=1} ^{l} \C{j} \cup (\C{l+1} \cup \S{l+1})\) is a \(\left( 1 \pm \varepsilon/(2k) \right)\)-spectral hypersparsifier of \(H_l'\).
For every vector \(\vect{x} \in \mathbb R^n\), we have
\begin{equation*}
    Q_{H_{l+1}'}(\vect{x}) 
    \leq 
    ( 1 + \varepsilon/(2k) ) Q_{H_{l}'}(\vect{x}) 
    \leq
    ( 1 + \varepsilon/(2k) ) 
    \left( 1 +  \varepsilon/(2k)  \right)^l Q_{H}(\vect{x})
    =
    \left( 1 +  \varepsilon/(2k)  \right)^{l+1} Q_{H}(\vect{x}).
\end{equation*}
Similarly, \(\left( 1 -  \varepsilon/(2k)  \right)^{l+1} Q_{H}(\vect{x}) \leq Q_{H_{l+1}'}(\vect{x})\), and so
\(H_{l+1}'\) is a \(\left( 1 \pm \varepsilon/(2k) \right)^{l+1}\)-spectral hypersparsifier of \(H\).

The desired bound follows from the fact that \(\ilast \leq k\), and
\begin{equation*}
    \left( 1 + \varepsilon/(2k) \right)^{k} \leq (1 + \varepsilon) \quad \text{and} \quad (1 - \varepsilon) \leq \left( 1 - \varepsilon/(2k) \right)^{k}.
\end{equation*}
\underline{Size of \(\widetilde H\):}
Let \(m_i\) be the number of hyperedges present in \S{i}, where \(1 \leq i \leq \ilast\).
We first show that \(m_{i} \leq 3 m_{i-1}/4\).
By \Cref{eq:coreset-probability}, \S{i} has size \atmost{ m_{i-1}/2 + \left( 2 c m_{i-1} \log m_{i-1} \right)^{1/2}}, so it suffices to show that  \(\left( 2 c m_{i-1} \log m_{i-1} \right)^{1/2} \leq m_{i-1}/4\).
By the assumption of the loop, we have
\begin{equation*}
    m_{i-1} \geq 32 c m^\star = 32 c n^2/\varepsilon^2 \log ^3m \geq 32 c \log m_{i-1},
\end{equation*}
and thus \(\left( 2 c m_{i-1} \log m_{i-1} \right)^{1/2} \leq m_{i-1}/4\).
This means that  the size of \S{\ilast} is 
\begin{equation*}
   \atmost{\max \{ \log m, m^\star \} } = \atmost{n^2/\varepsilon^2 \log ^3m}.
\end{equation*}

The rest of \(\widetilde H\) consists of the coresets \(\C{1}, \dots, \C{\ilast}\).
By \Cref{lem:coreset-and-sample}, each \C{i} has size \(\atmost{\lambda _i n^2} = n^2 / (\varepsilon / 2k)^2 \log^3 m\). 
Since \(k = \atmost{\log m}\), %
the total size of the coresets is
\begin{equation*}
    \atmost{\sum _{l = 1} ^k k^2 n^2 / \varepsilon^2 \log^3 m} = \atmost{n^2 / \varepsilon^2 \log^6 m}.
\end{equation*}

Therefore, the total size of \(\widetilde H\) is
\begin{equation*}
    \atmost{n^2/\varepsilon^2 \log ^3m + n^2 / \varepsilon^2 \log^6 m}
    =
    \atmost{n^2 / \varepsilon^2 \log^6 m}.
\end{equation*}

\underline{High probability claim:}
Since \(\ilast = \atmost{\log m}\) and from \Cref{lem:coreset-and-sample},
\(\widetilde H\) is a \SpectralHypersparsifier{} with probability at least \(1 - \atmost{\log m / m^c}\).
The claim follows by choosing \(c \geq 2\).
\end{proof}

\subsection{The Dynamic Algorithm} \label{subsec:dynamic}
In this section, we dynamize our static algorithm (\Cref{alg:spectral-sparsify}).
We first design a decremental data structure in \Cref{subsub:decremental}, where \(H\) undergoes only hyperedge deletions.
Then, we reduce it to a fully dynamic data structure in \Cref{subsub:fully-dynamic} using the following lemma.

\begin{lemma} \label{lem:reduction}
Assume that there is a decremental algorithm that with probability at least \(1 - 1/\poly{m'}\), maintains a \SpectralHypersparsifier{} of any hypergraph with \(m'\) initial hyperedges in \(T(m',n,\varepsilon^{-1})\) amortized update time  and of size \(S(m',n,\varepsilon^{-1})\), where \( S \) and \( T \) are  monotone non-decreasing functions. 
Then, the algorithm can be transferred into a fully dynamic algorithm that, with high probability, maintains a \SpectralHypersparsifier{} of any hypergraph with \(m\) hyperedges (at any point) in \atmost{T(m,n,\varepsilon^{-1}) \log m} amortized update time of size \atmost{S(m,n,\varepsilon^{-1})\log m}.
\end{lemma}

The fully dynamic to decremental reduction (\Cref{lem:reduction}) uses the batching technique \cite{Abraham:2016ab,Goranci:2025aa,Khanna:2025aa} and is explained in \Cref{subsub:fully-dynamic}.

\subsubsection{Decremental Spectral Sparsifier} \label{subsub:decremental}
In this section, we explain our data structure (\Cref{alg:decremental-spectral-sparsify}) to decrementally maintains a \SpectralHypersparsifier{} \(\widetilde H\) of \(H\).
As \Cref{alg:coreset-and-sample} is used as a subroutine of \Cref{alg:spectral-sparsify}, we begin by explaining its decremental implementation (presented in \Cref{alg:decremental-coreset-and-sample}).

\paragraph*{Decremental Implementation of \Cref{alg:coreset-and-sample}.}
To decrementally maintain a coreset \C{} and a sampled hypergraph \S{} of \(H\), we need to ensure that after each deletion, the maintained \(\C{} \cup \S{}\) remains a valid sparsifier for \(H\).
i.e., it continues to be a \SpectralHypersparsifier{} of \(H\) (\Cref{lem:coreset-and-sample}).

Recall that, for every pair \((u, v) \in V \times V\),
 \Cref{alg:coreset-and-sample} defines the set \E{u, v} of hyperedges where \(e \in \E{u,v}\) iff \(u \in \tail{e}\) and \(v \in \head{e}\).
It then constructs a coreset \C{} by adding \(\lambda\) (defined in \Cref{alg:coreset-and-sample}) heaviest hyperedges of \(\E{u,v} \setminus \C{}\) to \C{}. 
The hypergraph \S{} is then obtained by sampling each hyperedge in \(H \setminus \C{}\) with probability \(1/2\) while doubling its weight.

To construct \C{}, the algorithm greedily chooses a pair \((u, v)\), adds its heavyweight hyperedges to \C{}, and continues with another pair \((u', v')\).
Note that the order of choosing the pairs might affect the choice of hyperedges included in \C{}.
Nevertheless, the guarantee of the algorithm (\Cref{lem:coreset-and-sample}) is independent of this order.
For example, it does not matter whether \((u, v)\) is chosen first or \((u', v')\) is.
We will use this fact later to maintain a valid \C{} and \S{} after each deletion.

We use the same procedure to initialize the decremental implementation.
Since we would need to find the heaviest hyperedges in \(\E{u, v}\setminus \C{}\) after a deletion, we also order each \E{u,v}.

Suppose that hyperedge \(e\) has been removed from \(H\).
Then, \(e\) must belong to one of the three cases below, which together cover all hyperedges of \(H\).

\begin{itemize}
\item  
If \(e\) belongs to neither \C{} nor to \S{}, then its removal does not affect \(\C{} \cup \S{}\).
This is because \C{} still contains the heaviest hyperedges associated with each pair \((u, v) \in V \times V\), and each hyperedge in \S{} has been independently sampled.
    
\item 
If \(e\) belongs to \S{}, similar to the previous case, the sampled hypergraph \S{} after the removal of \(e\) is still a valid sample.
Thus, no further changes are required to maintain \(\C{} \cup \S{}\).
    
\item 
If \(e\) belongs to \C{}, then the deletion from \C{} translates to undoing the addition of \(e\) to \C{} from the specific set \E{u, v} that originally added \(e\) to \C{}.
In this case, since \(e\) no longer exists in \E{u, v}, we add a heaviest hyperedge \(e'\) from \(\E{u, v} \setminus \C{}\) to \C{}.
Since the order of hyperedge addition to \C{} does not affect its guarantees (\Cref{lem:coreset-and-sample}), the updated \C{} remains valid.
For bookkeeping reasons, if \(e'\) was previously sampled, we remove it from \S{}.
Based on the previous discussion, the updated \S{} is also valid. 
\end{itemize}

In addition to the changes explained above, the maintenance involves the removal of \(e\) from every set \E{u, v} containing \(e\), which, as explained in the lemma below, adds \atmost{\log m} overhead to the update time.
Putting it all together, \Cref{alg:decremental-coreset-and-sample} is our decremental data structure for maintaining \(\C{} \cup \S{}\).
We have the following.

\begin{algorithm}[t]
\KwIn{an \(m\)-edge hypergraph \(H = (V, E, \vect{w})\) and \( 0 < \varepsilon < 1 \)}
\Maintain{a coreset \C{} and a sampled hypergraph \S{} of \(H\)}
$C, S \gets (V, \emptyset, \vect{0})$\;

\Procedure{Initialize}{
\((C, S) \gets \textsc{Coreset-And-Sample(\( H, \varepsilon \))}\)
\tcc*[f]{initialize \(\C{}\) and \(\S{}\)}

    \ForEach(){pair \( (u, v) \in V \times V \)}{
        Order hyperedges in \E{u, v} from highest to lowest weight
    }
}

\Procedure{Delete(\(e\))}{
    Remove \(e\) from \(H\)
    
    \If{\C{} contains \(e\)}{
    Let \E{u, v} be the set from which \(e\) was added to \C{}\;
    Choose a hyperedge \(e'\) in \(\E{u, v} \setminus \C{}\) with highest weight and add it to \C{}\;
    Remove \(e'\) from \S{}\;
    Remove \(e\) from \C{}
    }
    
    \ForEach{pair \((u, v) \in V \times V\) with \(u \in \tail{e}\) and \(v \in \head{e}\)}{
        Remove \(e\) from \E{u, v}
    }
    
    Remove \(e\) from \S{}
}
\caption{Decremental-Coreset-And-Sample(\( H, \varepsilon \))}
\label{alg:decremental-coreset-and-sample}
\end{algorithm}

\begin{lemma} \label{lem:decremental-coreset-and-sample}
Given a constant \(c \geq 2\) and an \(m\)-edge \(n\)-vertex hypergraph \(H = (V, E, \vect{w})\) of rank \(r\) undergoing hyperedge deletions.
If \(m \geq c \log m\), then \Cref{alg:decremental-coreset-and-sample} maintains a \SpectralHypersparsifier{} \(\C{} \cup \S{}\) of \(H\) of size \atmost{m/2+ (2 c m \log m)^{1/2} + n^2 / \varepsilon ^2 \log ^3 m} in \atmost{r^2 \log m} amortized update time with probability at least \(1 - 1/m^{c-1}\).
\end{lemma}

\begin{proof}
Suppose the hyperedge \(e\) is removed from \(H\).

\underline{Correctness:} 
From the discussion above, it follows immediately that \(\C{} \cup \S{}\), after the update, is a valid sparsifier for \(H\); the correctness of its spectral property is addressed in the high probability claim below.

\underline{Size of \(\C{} \cup \S{}\):}
Since after each hyperedge deletion from \(\C{} \cup \S{}\), the algorithm substitutes it with at most one other hyperedge from \(H\), the size of \(\C{} \cup \S{}\) is monotonically decreasing.
Thus, \(\C{} \cup \S{}\) has the same size guarantee as of \Cref{alg:coreset-and-sample} explained in \Cref{lem:coreset-and-sample}, which is
\begin{equation*}
    \atmost{m/2 + (2 c m \log m)^{1/2} + \lambda n^2} = \atmost{m/2+ (2 c m \log m)^{1/2} + n^2 / \varepsilon ^2 \log ^3 m}.
\end{equation*}
    
\underline{Update time:} 
Since each hyperedge \(e\) contains \atmost{r} vertices, the total number of \E{\cdot, \cdot}'s containing \(e\) is \atmost{r^2}.
Thus, constructing \E{\cdot, \cdot}'s, %
takes \atmost{mr^2} time, while ordering them adds \atmost{mr^2 \log m} time to the initialization time as well.
As explained before, the deletion of \(e\) from \(H\) translates to its deletion from at most \(r^2\) sets of \E{\cdot, \cdot}'s, each of which takes \atmost{\log m} time, i.e., \atmost{mr^2 \log m} total time.
Other changes, such as substituting another hyperedge in \(S\), can be done by probing \E{\cdot, \cdot}'s only once in total.
We conclude that the total update time is \atmost{mr^2 \log m}, resulting in \atmost{r^2 \log m} amortized update time.

\underline{High probability claim:}
By choosing \(c \geq 2\) in \Cref{lem:coreset-and-sample}, each update is guaranteed to succeed with probability at least \(1 - 1/m^c\).
Since there are at most \(m\) updates, it follows that \Cref{alg:decremental-coreset-and-sample} succeeds with probability at least \(1 - 1/m^{c-1}\).%
\end{proof}

\paragraph*{Decremental Implementation of \Cref{alg:spectral-sparsify}.}
Our decremental data structure for maintaining a \SpectralHypersparsifier{} \(\widetilde H\) of \(H\) (\Cref{alg:decremental-spectral-sparsify}) is a decremental implementation of \Cref{alg:spectral-sparsify}.
Recall that, in \Cref{alg:spectral-sparsify}, we recurse on the sampled hypergraphs \(\S{1}, \dots, \S{\ilast}\) and add the coresets \(\C{1}, \dots, \C{\ilast}\) to \(\widetilde H\).
By \Cref{lem:spectral-sparsify}, \(\widetilde H = \bigcup_{j = 1} ^\ilast \C{j} \cup \S{\ilast}\) is a \SpectralHypersparsifier{} of \(H\).

\begin{algorithm}[t]
\KwIn{an \(m\)-edge hypergraph \(H = (V, E, \vect{w})\) and \( 0 < \varepsilon < 1 \)}
\Maintain{a \SpectralHypersparsifier{} \(\widetilde H\) of \(H\)}
\( i \gets 0 \)\;
\( k \gets \lceil \log _{3/4} m \rceil \)\;
\(m^\star \gets n^2/\varepsilon^2 \log ^3m\)\;
\( \S{0} \gets H \)\;

\Procedure{Initialize}{
    \While(\tcc*[f]{\(c\) is from \Cref{lem:coreset-and-sample}}){\( i \leq k \) and \( |E(H_i)| \geq 32 c m^\star \)}{
    
        \(\mathcal A_{i+1} \gets\) initialize \textsc{Decremental-Coreset-And-Sample(\( S_i, \varepsilon /(2k) \))} 
        
        \tcc*[f]{\(\mathcal A_{i+1}\) decrementally maintains \(\C{i+1}\) and \(\S{i+1}\)}
        
        \(\sparsifier{H} \gets \sparsifier{H} \cup \C{i+1}\)\;
        \( i \gets i + 1 \)
        }
     
    \( \ilast \gets i \)\;
    \(\sparsifier{H} \gets \sparsifier{H} \cup \S{\ilast}\)
}

\Procedure{Delete(\(e\))}{
    Remove \(e\) from \(H\) and \(\widetilde H\)\;
    \(i \gets 1\)\;
    
    \While{\(i \leq \ilast\)}{
        Pass the deletion of \(e\) to \(\mathcal A_i\)
        
        \If{\(e'\) has been added to \C{i} due to the deletion of \(e\)}{
            Add \(e'\) to \(\widetilde H\)\;
            Remove \(e'\) from \S{i}\;
            \(e \gets e'\)\;
        }
        
        \(i \gets i + 1\)\;
    }
}

\caption{Decremental-Spectral-Sparsify(\( H, \varepsilon \))}
\label{alg:decremental-spectral-sparsify}
\end{algorithm}

In \Cref{alg:decremental-spectral-sparsify}, we decrementally maintain the hypergraphs \(\S{1}, \dots, \S{\ilast}\) and the coresets \(\C{1}, \dots, \C{\ilast}\).
Suppose that hyperedge \(e\) has been deleted from \(H\) as the most recent update.
Below, we discuss how the data structure handles the deletion in all possible scenarios.
\begin{itemize}
\item  
If \(e\) does not belong to \(\widetilde H\), then
\(e\) can only exist in the sets of sampled hyperedges \(\S{1}, \dots, \S{\ilast - 1}\).
In this case, removing \(e\) from all \(\S{i}\)'s does not affect \(\widetilde H\) since the set of sampled hyperedges remains valid, as each hyperedge is sampled independently.

\item 
If \(e\) belongs to \(\S{\ilast}\), then \(e\) is present in all sets of sampled hyperedges, 
i.e., in \(\S{1}, \dots, \S{\ilast}\).
Note that, although \(e\) belongs to \(\widetilde H\) in this case, since it is only present in sampled hypergraphs, its removal does not affect the coresets and can be handled similarly to the previous case.

\item 
If \(e\) belongs to a coreset \C{i}, then it belongs to the sets of sampled hyperedges \(\S{1}, \dots, \S{i-1}\), and since it is in the coreset \C{i} of \S{i-1}, it cannot be present in \(\S{i}, \dots, \S{\ilast}\).
Similar to the previous cases, we can simply remove \(e\) from \(\S{1}, \dots, \S{i-1}\) to maintain \(\C{1} \cup \S{1}, \dots, \C{i-1} \cup \S{i-1}\).
For \(\C{i} \cup \S{i}\), we need to maintain \C{i}, for which we use \Cref{alg:decremental-coreset-and-sample} as a subroutine to find a hypergraph \(e'\) to be added to \C{i} as the substitution for \(e\).
By \Cref{alg:decremental-coreset-and-sample}, \(e'\) belongs to \S{i-1} and thus
its addition to \C{i} does not affect \(\C{1} \cup \S{1}, \dots, \C{i-1} \cup \S{i-1}\).
On the other hand, \(e'\) might be present in \(\C{i+1} \cup \S{i+1}, \dots, \C{\ilast} \cup \S{\ilast}\), and we need to ensure they are still valid after adding \(e'\) to \C{i}.
If \(e'\) appears in the sets of sampled hyperedges, then similar to the previous cases, simply removing \(e'\) from the sparsifiers makes them valid.
Therefore, we pass the deletion of \(e'\) to the next sparsifiers until we reach \(\C{j} \cup \S{j}\) containing \(e'\) in its coreset \C{j}.
Again, by construction, \(e'\) cannot be present in \(\C{j+1} \cup \S{j+1}, \dots, \C{\ilast} \cup \S{\ilast}\), and we need to remove \(e'\) from \C{j}.
This procedure is similar to the one we just explained for the removal of \(e\) from \C{i}, except that we now use it to remove \(e'\) from \C{j}.
\end{itemize}

As explained above, to handle hyperedge deletions in each sparsifier \(\C{i} \cup \S{i}\), the data structure utilizes \Cref{alg:decremental-coreset-and-sample} and then transmits the changes in \(\C{i} \cup \S{i}\) to \(\C{i+1} \cup \S{i+1}\).
The key point of our data structure is that it guarantees the transfer of at most \textit{one} hyperedge deletion from one level to the next, thereby ensuring low recourse at each level.
This results in at most \(\ilast = \atmost{\log m}\) hyperedge deletions across all levels following a hyperedge deletion in \(H\).
The guarantees of the data structure are stated below.

\begin{lemma} \label{lem:decremental-spectral-sparsify}
Given a constant \(c \geq 3\) and an \(m\)-edge \(n\)-vertex hypergraph \(H = (V, E, \vect{w})\) of rank \(r\) undergoing hyperedge deletions, \Cref{alg:decremental-spectral-sparsify} maintains a \SpectralHypersparsifier{} \(\widetilde H\) of \(H\) of size \atmost{n^2 / \varepsilon ^2 \log ^6 m} in \atmost{r^2 \log ^2 m} amortized update time with probability at least \(1 - 1/m^{c-2}\).
Additionally, each deletion in \(H\) results in \atmost{\log m} recourse in \(\widetilde H\).
\end{lemma}

\begin{proof}
Suppose that a deletion has happened in \(H\).

\underline{Correctness:}
From \Cref{lem:decremental-coreset-and-sample}, each \(\mathcal A_i\) correctly handles the deletions.
The fact that the data structure correctly transfers hyperedge deletions between \(\mathcal A_i\)'s and thus correctly maintains \(\widetilde H\) follows from the discussion above.

\underline{Size of \(\widetilde H\):}
The data structure includes \(\ilast = \atmost{\log m}\) coresets \(\C{1}, \dots, \C{\ilast}\) in \(\widetilde H\).
By \Cref{lem:decremental-coreset-and-sample}, for each \(1 \leq i \leq \ilast\), 
\C{i} has a size of \atmost{n^2/\varepsilon _i ^{2} \log ^3 m}, where \(\varepsilon _i = \varepsilon / (2k)\) and \(k = \atmost{\log m}\).
As shown in the proof of \Cref{lem:spectral-sparsify}, the size of \S{\ilast} is \atmost{n^2/\varepsilon^2 \log ^3m}.
It follows that the size of \(\widetilde H\) is
\begin{equation*}
    \atmost{ \sum_{i = 1} ^k n^2/\varepsilon _i ^{2} \log ^3 m + n^2/\varepsilon^2 \log ^3m} = \atmost{n^2/\varepsilon ^{2} \log ^6 m}.
\end{equation*}

\underline{Update time:} 
The data structure initializes and decrementally maintains 
\(\ilast = \atmost{\log m}\) data structures of \Cref{alg:decremental-coreset-and-sample}.
By \Cref{lem:decremental-coreset-and-sample}, each data structure takes \atmost{m r^2 \log m} total update time.
Also, at each update, the data structure transmits  at most one hyperedge  from one level to the next, which results in an \atmost{m \log m} total transmission.
Thus, the algorithm takes \atmost{m r^2 \log^2 m} total update time, or equivalently,
\atmost{r^2 \log^2 m} amortized update time.

\underline{High probability claim:}
From \Cref{lem:decremental-coreset-and-sample}, for any constant \(c \geq 2\), each \(\mathcal A_i\) correctly maintains its coreset and the set of sampled hyperedges with probability at least \(1 - 1/m^{c - 1}\).
Since \(1 \leq i \leq \log m\), by choosing \(c \geq 3\), the claim follows.

\underline{Recourse bound:} %
As explained before, every hyperedge deletion from \(H\) translates to at most one hyperedge deletion or addition in \(\mathcal A_i\) for \(1 \leq i \leq \ilast\).
Since \(\ilast = \atmost{\log m}\), it follows that the total number of hyperedge changes in \(\widetilde H\) is \atmost{\log m}.
\end{proof}

\subsubsection{Fully Dynamic Spectral Sparsifier} \label{subsub:fully-dynamic}

With a decremental data structure (\Cref{alg:decremental-spectral-sparsify}) in hand, we now obtain the fully dynamic data structure.
We begin by explaining the reduction from the fully dynamic to the decremental data structure and by proving \Cref{lem:reduction}.
We then prove \Cref{th:main}.

\paragraph*{Reduction from Fully Dynamic to Decremental.}
We explain how to use the decremental data structure to design a fully dynamic data structure, as described in \Cref{alg:reduction}.
In a nutshell, the data structure uses the batching technique; it maintains a batch of \((1 \pm \varepsilon)\)-spectral hypersparsifiers \(\widetilde H_1, \dots, \widetilde H_k\), each of which decrementally maintains sub-hypergraphs \(H_1, \dots, H_k\), respectively, such that these sub-hypergrphs partition \(H\).
The union \(\widetilde H = \bigcup_{i = 1} ^k \widetilde H_i\), by decomposability (\Cref{lem:decomposability}), is a \SpectralHypersparsifier{} of \(H\).

\begin{algorithm}[t]
\KwIn{an empty \(n\)-vertex hypergraph \(H = (V, E, \vect{w})\) with \(|E| \leq m\) at any time}
\Maintain{a \SpectralHypersparsifier{} \(\widetilde H\) of \(H\)}
\Procedure{Initialize}{
    \ForEach{\(1 \leq i \leq \log m\)}{
        \(H_i, \widetilde H_i \gets (V, \emptyset, \vect{0})\)
    }
    
    \(t \gets 0\)\;
    \(\ilast \gets 1\)\;
}

\Procedure{Add(\(e\))}{
    \(t \gets t+1\)\;
    \(j \gets \max\{i \mid \text{\(t\) is divisible by \(2^{i-1}\)}\}\)\;
    \(\ilast \gets \max\{ j, \ilast \}\)\;
    \(H_{j} \gets H_{j} \cup H_{j-1} \cup \dots \cup H_1\)\;
    Add \(e\) to \(H_{j}\)\;
    
    \ForEach{\(1 \leq i \leq j - 1\)}{
        \(H_i, \widetilde H_i \gets (V, \emptyset, \vect{0})\)\;
    }
    
    \(\mathcal A_{j} \gets\) (re)initialize \textsc{Decremental-Spectral-Sparsify(\( H_{j}, \varepsilon \))}\; 
    \tcc*[f]{\(\mathcal A_{j}\) decrementally maintains the \textit{newly} computed \(\widetilde H_{j}\)}\;
}

\Procedure{Delete(\(e\))}{
    \(j \gets \text{index of the specific hypergraph among \(H_1, \dots, H_n\) that contains \(e\)}\)\;
    Pass the deletion of \(e\) to \(\mathcal A_j\)\;
}

\Procedure{OutputSparsifier()}{
    \Return \(\widetilde H = \bigcup_{i = 1} ^ \ilast \widetilde H_i\)\; 
}

\caption{Fully-Dynamic-Spectral-Sparsify(\(H, \varepsilon\))}
\label{alg:reduction}
\end{algorithm}

More specifically, the data structure starts with an empty \(H\) and empty \(H_1, \dots, H_k\) and ensures that each \(H_i\) contains at most \(2^i\) hyperedges at all time (which we call the size constraint of \(H_i\)).
Intuitively, a hyperedge \(e\) inserted in \(H\) is placed in \(H_1\) as long as the size constraint of \(H_1\) is not violated.
If there are already two hyperedges in \(H_1\), the data structure moves \textit{all} hyperedges of \(H_1\) (along with \(e\)) to the empty \(H_2\).
Note that \(H_2\) can contain at most four hyperedges, and so its size constraint is not violated.
However, if there were already two hyperedges in \(H_2\), then the addition of three more hyperedges (from \(H_1\) plus \(e\)) would violate its size constraint.
In this case, the data structure would move the hyperedges of \(H_1\) and \(H_2\) plus \(e\) to \(H_3\), and so on.

To regularize the insertion process, we initialize a counter \(t\) as a binary sequence of \(\log k\) zeros.
After each insertion, the data structure increments \(t\) (in the binary format) by one and finds the highest index \(i\) whose bit flipped due to the increment.
The data structure then moves the hyperedges in \(H_1, \dots, H_{i-1}\), along with the inserted hyperedge \(e\), to \(H_i\) and reinitializes the decremental data structure with the new \(H_i\).

The deletion process can be done as before; since each hyperedge is associated with exactly one sub-hypergraph \(H_i\), its deletion would be easily handled by passing it to the decremental data structure maintaining \(\widetilde H_i\).

\begin{proof}[Proof of \Cref{lem:reduction}]
Suppose that the decremental data structure maintains a \SpectralHypersparsifier{} of any hypergraph with \(m'\) initial hyperedges in \(T(m',n,\varepsilon^{-1})\) amortized update time and of size \(S(m',n,\varepsilon^{-1})\) with probability at least \(1 - 1/\poly{m'}\).

\underline{Correctness:}
From the discussion above, the sub-hypergraphs \(H_1, \dots, H_k\) correctly partition the hyperedges of \(H\).
By \Cref{lem:decremental-spectral-sparsify}, the sparsifiers \(\widetilde H_1, \dots, \widetilde H_k\) are correctly maintained.
It follows from \Cref{lem:decomposability} that \(\widetilde H = \bigcup_{i = 1} ^k \widetilde H_i \) is a \SpectralHypersparsifier{} of \(H\).

\underline{Size of \(\widetilde H\):}
By assumption, each \(\widetilde H_i\) has size \(S(m_i,n,\varepsilon^{-1})\), which is upper bounded by \(S(m,n,\varepsilon^{-1})\) since \(m_i \leq m\) and \(S\) is a  monotone non-decreasing function.
Since \(k = \atmost{\log m} \), the size of \(\widetilde H\) is bounded by \(\sum _{i = 1} ^{\log m} S(m_i,n,\varepsilon^{-1}) = \atmost{S(m,n,\varepsilon^{-1}) \log m}.\)

\underline{Update time:}
After \(l\) insertions, \(H_i\) has been reinitialized for at most \(l/2^i\) times.
This is because \(H_i\) can contain at most \(m_i = 2^i\) hyperedges.
By assumption, the total time for the initialization and maintenance of \(\mathcal A_i\) is bounded by \(m_i T(m_i,n,\varepsilon^{-1}) = 2^i T(m_i,n,\varepsilon^{-1})\), which is bounded by \(2^i T(m,n,\varepsilon^{-1})\) since \(m_i \leq m\) and \(S\) is a  monotone non-decreasing function.
Thus, the total time for the initialization and maintenance of \(\mathcal A_i\) throughout \(l\) insertions is \(\atmost{\frac{l}{2^i} 2^i T(m,n,\varepsilon^{-1}) } = \atmost{l T(m,n,\varepsilon^{-1}) }\).
Therefore, the total update time for maintaining \(\widetilde H_1, \dots, \widetilde H_k\) is
\begin{equation*}
\atmost{\sum_{i = 1} ^{\log m} l T(m,n,\varepsilon^{-1})} = \atmost{l T(m,n,\varepsilon^{-1}) \log m},
\end{equation*}
resulting in an \atmost{T(m,n,\varepsilon^{-1}) \log m} amortized update time.

\underline{High probability claim:}
After \(l\) insertions, each \(H_i\) is reinitialized \atmost{l} times.
By \Cref{lem:decremental-spectral-sparsify}, the probability of failure is at most \(l/m^{c-2}\), where \(c \geq 3\).
Since \(k = \log m\), it follows that the probability of failure for maintaining \(\widetilde H_1, \dots, \widetilde H_k\) is at most \(l \log m/m^{c-2}\).
Since \(l = \poly{m}\), the high probability claim follows by choosing \(c\) to be large enough.
\end{proof}

\begin{proof}[Proof of \Cref{th:main}]
Given an \(m\)-edge \(n\)-vertex hypergraph \(H \) of rank \(r\), by \Cref{lem:decremental-spectral-sparsify},
the decremental data structure (\Cref{alg:decremental-spectral-sparsify}) maintains a \SpectralHypersparsifier{} of \(H\) of size \(S(m',n,\varepsilon^{-1}) = \atmost{n^2 / \varepsilon ^2 \log ^6 m} \) in \(T(m,n,\varepsilon^{-1}) = \atmost{r^2 \log ^2 m} \) amortized update time.

Since \( S \) and \( T \) are  monotone non-decreasing functions, it follows from \Cref{lem:reduction} that, for any hypergraph \(H\) containing at most \(m\) hyperedges at any point, the fully dynamic data structure (\Cref{alg:reduction}) maintains a \SpectralHypersparsifier{} of \(H\) of size \atmost{n^2 / \varepsilon ^2 \log ^7 m} in \atmost{r^2 \log ^3 m} amortized update time.
\end{proof}

\section*{Acknowledgement}
We thank the anonymous reviewers for their helpful suggestions.

SF: This project has received funding from the European Research Council (ERC) under the European Union's Horizon 2020 research and innovation programme (grant agreement No 947702).

\clearpage

\section*{References}
\printbibliography[heading=none]

\clearpage

\newpage
\appendix
\section{Parallel Batch-Dynamic Spectral Sparsification} \label{sec:parallel}
In this section,
we present our batch parallel algorithm
for maintaining a \SpectralHypersparsifier{} \(\widetilde H\)
of a directed hypergraph \(H = (V, E, \vect{w})\)
undergoing hyperedge deletions and additions.

The algorithm is an extension of 
our fully dynamic sequential data structure (\Cref{alg:reduction}).
Recall that to maintain \(\widetilde H\) in the fully dynamic setting,
\Cref{alg:reduction} decrementally maintains 
a sequence of sub-hypergraphs \(H_1, \cdots, H_\ilast\) using \Cref{alg:decremental-spectral-sparsify},
which itself builds upon \Cref{alg:decremental-coreset-and-sample}.
Therefore, to explain our fully dynamic extension,
we begin by explaining how we 
implement \Cref{alg:decremental-coreset-and-sample,alg:decremental-spectral-sparsify} in the batch parallel setting.

\subsection{Parallel Batch-Dynamic Implementation of \Cref{alg:decremental-coreset-and-sample}.}
We first analyze each step of the algorithm individually.
In the end, we will combine them 
to obtain the guarantees of the algorithm.

The initialization process of \Cref{alg:decremental-coreset-and-sample}
consists of finding a coreset \C{} 
and a sampled hypergraph \S{}.
To handle this process, for every pair \(u,v \in V\),
we defined \E{u,v} to contain a hyperedge \(e\)
iff \(u \in \tail{e}\) and \(v \in \head{e}\).
We explain in the claim below how to construct \E{\cdot, \cdot}'s in parallel.
\begin{claim} \label{cl:parallel-E-construction}
All sets \E{\cdot, \cdot} can be constructed in \atmost{mr^2 \log m} total work and \atmost{\log r} depth.
\end{claim}
\begin{proof}
Since each hyperedge \(e\)
consists of \atmost{r} vertices,
it can appear in at most \(r^2\) of the sets \E{\cdot, \cdot}'s.
For every hyperedge \(e\), 
we identify the corresponding sets \E{\cdot, \cdot}'s and add \(e\) to them.
Since there are initially $m$ hyperedges and adding them to each \E{\cdot, \cdot} takes \atmost{\log m} work, 
the bound on the total work follows.

The bound on depth follows from
the following observation.
By definition,
to construct \E{\cdot, \cdot}'s,
we need to find the pairs \((u, v)\) for every hyperedge \(e\)
with \(u \in \tail{e}\) and \(v \in \head{e}\)
and add \(e\) to \E{u, v}.
Although \tail{e} and \head{e} have size at most \(r\),
the process can be divided into similar tasks 
on subsets of \tail{e} and \head{e} with size at most \(r/2\).
More precisely,
we can partition \(\tail{e} = A \cup B\) and  \(\head{e} = C \cup D\)
so that \(A, B, C, D\) all have size at most \(r/2\).
To find the pairs associated with \(e\),
we can find the pairs in the four subcases corresponding to \((A, C), (A, D), (B, C), (B, D)\).
Thus, \(D(r) = 4 D(r/2)\), which solves to \atmost{\log r} depth.
\end{proof}

Once the \E{\cdot, \cdot}'s are computed,
the algorithm chooses \(\lambda\) heavyweight hyperedges from each \(\E{\cdot,\cdot} \) 
to be included in coreset \C{}.
Since we look for the heaviest hyperedges in each set,
we sort them from heaviest to lightest
so that they can be accessed more efficiently after an update.
This leads to the following claim.

\begin{claim} \label{cl:parallel-E-sorting}
All sets \E{\cdot, \cdot} can be sorted 
in \atmost{m r^2 \log m} total work 
and \atmost{\log ^2 m} depth.
\end{claim}
\begin{proof}
For each \E{\cdot, \cdot}, 
we use the parallelized version of \textsc{MergeSort},
which has \arxivVsConference{}{\linebreak} \atmost{|\E{\cdot,\cdot}| \log |\E{\cdot,\cdot}|} work 
and \atmost{\log ^2 |\E{\cdot,\cdot}|} depth. 
Thus, the depth is \atmost{\log ^2 m},
and the total work is
\begin{equation*}
\atmost{\sum _{u,v \in V} |\E{u,v}| \log |\E{u,v}|} = \atmost{\sum _{u,v \in V} |\E{u,v}| \log m} = \atmost{mr^2 \log m},
\end{equation*}
following from the fact that
each \E{\cdot, \cdot} contains \atmost{m} hyperedges,
and that the total number of hyperedges in \E{\cdot, \cdot}'s is \atmost{mr^2}.
\end{proof}

As mentioned above,
the coreset \C{} is computed 
by choosing \(\lambda\) (if available) hyperedges in each \(\E{\cdot, \cdot} \setminus \C{}\)
and adding them to \C{}.
We discuss this process in the following claim.

\begin{claim} \label{cl:parallel-coreset}
The coreset \C{} can be computed
in \atmost{\lambda n^2 \log m} total work
and \atmost{\log m} depth.
\end{claim}
\begin{proof}
The algorithm chooses at most \(\lambda\) hyperedges from each \(\E{\cdot, \cdot} \setminus \C{}\).
Since each \E{\cdot,\cdot} is sorted,
the algorithm only needs to probe the set
and add a hyperedge if it has not already been added to \C{} by another processor.
However, the algorithm might probe \atmost{\lambda n^2} of the hyperedges in \E{\cdot, \cdot},
which have already been added to \C{} by other processors.
Thus, it takes \(\atmost{\lambda n^2 \log m + \lambda \log m} = \atmost{\lambda n^2 \log m}\) total work for each \E{\cdot, \cdot} to add enough edges to \C{},
where the \atmost{\log m} overhead 
is due to standard techniques for avoiding memory conflicts. %
\end{proof}

After computing the coreset hypergraph \C{},
the sampled hypergraph \S{} is computed 
by sampling each hyperedge in \(H \setminus \C{}\)
with probability \(1/2\)
and doubling its weight.
This process can simply be parallelized
as explained in the following claim.

\begin{claim} \label{cl:parallel-sample}
The sampled hypergraph \S{} 
can be computed in \atmost{m} total work
and \atmost{1} depth.
\end{claim}
\begin{proof}
Since the sampling is independent for each hyperedge,
the algorithm can sample each of them separately.
The guarantees follow from the fact 
that the sampling of each hyperedge can be done 
in \atmost{1} time and \atmost{1} depth.
\end{proof}

We now proceed to explain how the algorithm handles deletions.
After the deletion of a hyperedge \(e\) from \(H\),
if \(e\) was included in \C{},
the algorithm attempts to substitute it with another hyperedge.
This process translates to
finding the specific set \E{u, v} that added \(e\) to \C{},
and then finding a heaviest hyperedge in \(\E{u,v}\setminus \C{}\) and adding it to \C{}.
In the following claim,
we discuss how this can be done in parallel when there is a batch of \(k\) hyperedge deletion.

\begin{claim} \label{cl:parallel-one-level-deletion}
Each batch of \(k\) hyperedge deletions can be handled 
in \atmost{k r^2\log m} amortized work
and \atmost{\log m} depth.
\end{claim}
\begin{proof}
As discussed in the sequential algorithm,
each process for substituting a deleted hyperedge
consists of \atmost{1} amortized changes in \C{} and \S{},
each of which costs \atmost{\log m} amortized work. 
To keep the depth short,
the algorithm must not choose an already deleted hyperedge as the substitution,
as this could result in \atmost{k} iterations to find the substitution.
To alleviate this issue,
we first remove the \(k\) hyperedges from \E{\cdot, \cdot}'s.
This results in \atmost{kr^2 \log m} extra work
since each hyperedge could be present in \atmost{r^2} sets,
and it takes \atmost{\log m} to remove an element from the sets.

Since updating \E{\cdot,\cdot}'s dominates the update time,
the amortized work of \atmost{kr^2 \log m} follows.
The bound on depth follows accordingly, 
as there is no dependency between the deletions.
\end{proof}

To obtain the guarantees of the entire algorithm,
we use the claims above 
and derive the following lemma. 
The proof follows directly
from the proofs of the claims, so we will not repeat them.

\begin{lemma}
Given an \(m\)-edge \(n\)-vertex hypergraph \(H = (V, E, \vect{w})\) of rank \(r\) 
undergoing batches of \(k\) hyperedge deletions,
the parallel batch-dynamic implementation of \Cref{alg:decremental-coreset-and-sample}
initiates and maintains
the coreset hypergraph \C{} 
and the sampled hypergraph \S{} 
of \(H\), with high probability,
in \atmost{k r^2 \log m} amortized work
and \atmost{\log ^2 m} depth.
\end{lemma}

\subsection{Parallel Batch-Dynamic Implementation of \Cref{alg:decremental-spectral-sparsify}.} \label{subsec:parallel-decremental-spectral-sparsify}
The data structure of \Cref{alg:decremental-spectral-sparsify}
recursively initializes \Cref{alg:decremental-coreset-and-sample}
on smaller hypergraphs 
for \(\ilast = \atmost{\log m}\) times
to build in the sequence of coresets \(\C{1}, \dots, \C{\ilast}\)
and the sequence of sampled hypergraphs \(\S{1}, \dots, \S{\ilast}\).
The initialization guarantees are discussed in the following claim.

\begin{claim}
\Cref{alg:decremental-spectral-sparsify} can be initialized
in \atmost{m r^2 \log m} total work
and \atmost{\log ^2 m} depth.
\end{claim}
\begin{proof}
We separately analyze the two subroutines 
involved in the initialization.
Combining their guarantees results in 
\atmost{mr^2 \log m} total work 
and \atmost{\log ^2 m} depth.

\underline{Constructing and sorting \E{\cdot,\cdot}'s:}
By \Cref{cl:parallel-E-construction,cl:parallel-E-sorting},
it takes \atmost{m r^2 \log m} total work and \atmost{\log ^2 m} depth 
to construct and sort all \E{\cdot, \cdot}'s.
The algorithm then uses these sets 
to find the coreset hypergraphs \(\C{1}, \dots, \C{\ilast}\)
and the sampled hypergraphs \(\S{1}, \dots, \S{\ilast}\).

\underline{Constructing the coreset and sampled hypergraphs:}
By the proof of \Cref{lem:spectral-sparsify},
with high probability, each \(\C{i} \cup \S{i}\) consists of at most \(3/4\) hyperedges compared to \(\C{i-1} \cup \S{i-1}\). %
Combining this fact with \Cref{cl:parallel-coreset,cl:parallel-sample},
it follows that the total time to construct the whole sequence is \(\atmost{\sum_{i = 1} ^\ilast \left(  3m/4 \right) ^i } = \atmost{m}\).
Since \(\C{i} \cup \S{i}\) is constructed on top of \(\C{i-1} \cup \S{i-1}\),
by \Cref{cl:parallel-coreset,cl:parallel-sample} and the fact that \(\ilast = \atmost{\log m}\),
we conclude that the depth of this process is \atmost{\log ^2 m}.
\end{proof}

Recall that \Cref{alg:decremental-spectral-sparsify} handles each hyperedge deletion
by maintaining \C{i} and \S{i} for every \(1 \leq i \leq \ilast\),
and by passing at most one hyperedge deletion from level \(i\) to level \(i+1\). 
Unfortunately, the inter-level hyperedge is level-dependent, and inherently sequential.
The following claim 
discusses how the algorithm
handles a batch of \(k\) hyperedge deletions.

\begin{claim}
Each batch of \(k\) hyperedge deletions
can be handled 
in \atmost{k r^2 \log ^2 m} amortized work
and \atmost{\log ^2 m} depth.
\end{claim}
\begin{proof}
By \Cref{cl:parallel-one-level-deletion}, 
each batch deletion at level \(i\) 
can be handled in \atmost{k r^2 \log m} amortized work and \atmost{\log m} depth.
Since each level depends on the deletion from the previous level, 
and each level handles at most \(k\) hyperedge deletions,
it follows that the total work is \atmost{k r^2 \log ^2 m},
and the depth is \atmost{\log ^2 m}.
\end{proof}

We summarize the guarantees of the entire algorithm 
in the lemma below. 
The proof follows directly 
from the claims above
and is therefore omitted.

\begin{lemma} \label{lem:parallel-decremental-spectral-sparsify}
Given an \(m\)-edge \(n\)-vertex hypergraph \(H = (V, E, \vect{w})\) of rank \(r\) 
undergoing batches of \(k\) hyperedge deletions,
the parallel batch-dynamic implementation of \Cref{alg:decremental-spectral-sparsify}
initiates and maintains
a \SpectralHypersparsifier{}, with high probability, 
in \atmost{k r^2 \log ^2 m} amortized work
and \atmost{\log ^2 m} depth.
\end{lemma}

\subsection{Parallel Batch-Dynamic Implementation of \Cref{alg:reduction}.}

Our batch parallel data structure
is based on \Cref{alg:reduction},
where its subroutine (\Cref{alg:decremental-spectral-sparsify}) is replaced with its batch parallel implementation (explained in \Cref{subsec:parallel-decremental-spectral-sparsify}).
For brevity, we do not repeat them here
and only highlight the differences in our implementation compared to \Cref{alg:reduction}.

Our batch parallel implementation
handles deletions similarly to \Cref{alg:reduction},
but now distributes up to \(k\) hyperedge deletions (one batch) across the sub-hypergraphs \(H_1, \dots, H_\ilast\),
 instead of distributing only a single hyperedge deletion. 

The insertions are handled similarly to the insertion process in \Cref{alg:reduction}.
The only difference is that
after each batch of insertions,
we increment the timer \(t\) by the number of inserted hyperedges (in the binary format), rather than just one.
The rest of the algorithm remains the same: we reinitialize the sub-hypergraph \(H_j\) to contain all the hyperedges from \(H_1, \dots, H_j\), along with the newly inserted batch of hyperedges,
where \(j\) is chosen based on \(t\) (see \Cref{alg:reduction} for more details).

We conclude this section by proving \Cref{th:parallel} below.

\begin{proof}[Proof of \Cref{th:parallel}.]
The proof of correctness, size of \(\widetilde H\), and high probability claim
directly follows from the proof of \Cref{lem:reduction}
and are not restated here.
Below, we prove how the data structure achieves an \atmost{kr^2 \log ^3 m} amortized work and \atmost{\log ^2 m} depth
by separately analyzing batch deletions and insertions.

\underline{Handling batch deletions:}
We pass the deletion of hyperedges
to their associated sub-hypergraph.
By \Cref{lem:parallel-decremental-spectral-sparsify},
each sub-hypergraph can handle a batch of \(k\) hyperedge deletions
in \atmost{k r^2 \log ^2 m} amortized work and \atmost{\log ^2 m} depth.
Since there are \atmost{\log m} sub-hypergraphs,
and that there is no dependency between sub-hypergraphs,
it follows that the fully dynamic data structure
can handle the deletions in \atmost{k r^2 \log ^3 m} amortized work and \atmost{\log ^2 m} depth.

\underline{Handling batch insertions:}
Using a similar approach to that used in the update time analysis of \Cref{lem:reduction},
it follows that
our implementation would have \atmost{\log m} overhead in work.
Combined with \Cref{lem:parallel-decremental-spectral-sparsify},
this results in an \atmost{kr^2 \log ^3 m} amortized work.
Since the sub-hypergraphs partition \(H\)
and there is no dependency between them, it follows that the depth remains \atmost{\log ^2 m}.
\end{proof}

\end{document}